\documentclass[review]{elsarticle}

\journal{Journal of \LaTeX\ Templates}

\usepackage[utf8]{inputenc}
\usepackage[T1]{fontenc}
\usepackage{tcolorbox}
\usepackage{amsmath,amsfonts,amssymb}
\usepackage{bm}
\usepackage{subfigure}
\usepackage{amsthm}
\usepackage{pifont}
\usepackage[colorinlistoftodos,prependcaption,textsize=tiny,]{todonotes}
\usepackage{algorithm,algpseudocode,float}

\usepackage{graphicx}

\definecolor{mypink2}{RGB}{204, 119,34}
\usepackage{tcolorbox}
\usepackage{xcolor}

\usepackage[colorlinks=true,urlcolor=blue,citecolor=blue,linkcolor=red!50!black,bookmarks=true]{hyperref}

\newcommand{\lleft}{\mathsf{left}}
\newcommand{\rright}{\mathsf{right}}
\newcommand{\mbs}{\mathsf{MBS}}
\newcommand{\mis}{\mathsf{MIS}}
\newcommand{\oct}{\mathsf{OCT}}
\newcommand{\fvs}{\mathsf{FVS}}
\newcommand{\opt}{\mathsf{OPT}}
\newcommand{\ptas}{\mathsf{PTAS}}
\newcommand{\np}{\mathsf{NP}}

\newcommand{\tfs}{\mathsf{MTFS}}

\newcommand{\apx}{\mathsf{APX}}
\newcommand{\w}{\mathsf{W[1]}}

\newcommand{\bsu}{\mathsf{BSUDG}}

\newcommand{\rectangle}{{%
	\ooalign{$\sqsubset\mkern3mu$\cr$\mkern3mu\sqsupset$\cr}%
}}

\usepackage{pifont}
\usepackage{dsfont}
\usepackage[colorinlistoftodos,prependcaption,textsize=tiny,]{todonotes}
\usepackage{algorithm,algpseudocode,float}

\usepackage[framemethod=tikz]{mdframed}
\newmdenv[backgroundcolor=black!07,
topline=false,
bottomline=false,
rightline=false,
skipabove=\topsep,
skipbelow=\topsep
]{siderules}

\definecolor{mypink2}{RGB}{204, 119,34}
\definecolor{azure}{RGB}{0,127,255}
\usepackage{tcolorbox}
\usepackage{xcolor}

\usepackage[nameinlink, capitalise]{cleveref}

\usepackage{mathtools}

\DeclarePairedDelimiter\floor{\lfloor}{\rfloor}

\newtheorem{observation}{Observation}
\newtheorem{theorem}{Theorem}

\newtheorem{lemma}[theorem]{Lemma}

\newtheorem{corollary}[theorem]{Corollary}

\begin{document}

\begin{frontmatter}
\title{Bipartite Subgraph of Geometric Intersection Graphs \tnoteref{mytitlenote}}
\tnotetext[mytitlenote]{A preliminary version of this paper appeared in the The 14th International Conference and Workshops on Algorithms and Computation (WALCOM 2020).}


\author[mymainaddress]{Satyabrata Jana\corref{mycorrespondingauthor}}
\cortext[mycorrespondingauthor]{Corresponding author}
\ead{satyamtma@gmail.com}
\author[mysecondaryaddress]{Anil Maheshwari}
\ead{anil@scs.carleton.ca}
\author[mysecondaryaddress]{Saeed Mehrabi}
\ead{saeed.mehrabi@carleton.ca}
\author[mymainaddress]{Sasanka Roy}
\ead{sasanka.ro@gmail.com}

\address[mymainaddress]{Indian Statistical Institute, Kolkata, India }
\address[mysecondaryaddress]{School of Computer Science, Carleton University, Ottawa, Canada}

\begin{abstract}
We study the \emph{Maximum Bipartite Subgraph} ($\mbs$) problem, which is defined as follows. Given a set $S$ of $n$ geometric objects in the plane, we want to compute a maximum-size subset $S'\subseteq S$ such that the intersection graph of the objects in $S'$ is bipartite. We first give a simple $O(n)$-time algorithm that solves the $\mbs$ problem on a set of $n$ intervals. We also give an $O(n^2)$-time algorithm that computes a near-optimal solution for the problem on circular-arc graphs. We show that the $\mbs$ problem is $\np$-hard on geometric graphs for which the maximum independent set is $\np$-hard (hence, it is $\np$-hard even on unit squares and unit disks). On the other hand, we give a $\ptas$ for the problem on unit squares and unit disks. Moreover, we show fast approximation algorithms with small-constant factors for the problem on unit squares, unit disks and unit-height rectangles. Finally, we study a closely related geometric problem, called \emph{Maximum Triangle-free Subgraph} ($\tfs$), where the objective is the same as that of $\mbs$ except the intersection graph induced by the set $S'$ needs to be triangle-free only (instead of being bipartite).
\end{abstract}
\end{frontmatter}
	

\section{Introduction}
\label{sec:introduction}
In this paper, we study the following geometric problem. Given a set $S$ of $n$ geometric objects in the plane, we are interested in computing a maximum-size subset $S'\subseteq S$ such that the intersection graph induced by the objects in $S'$ is bipartite. We refer to this problem as the \emph{Maximum Bipartite Subgraph} ($\mbs$) problem. The $\mbs$ problem is closely related to the Odd Cycle Transversal ($\oct$) problem: given a graph $G$, the objective of the $\oct$ problem is to compute a minimum-cardinality subset of $S\subseteq V(G)$ such that the intersection of $S$ and the vertices of every odd cycle of the graph is non-empty. Notice that $\mbs$ and $\oct$ are equivalent for the class of graphs on which $\oct$ is polynomial-time solvable: an exact solution $S$ for $\oct$ gives $V(G)\setminus S$ as an exact solution for $\mbs$ within the same time bound (see below for a summary of the main known results on the $\oct$ problem). However, on classes of graphs for which $\oct$ is $\np$-hard, an $\alpha$-approximation algorithm for $\oct$ might not provide any information on the approximability of $\mbs$ on the same classes of graphs.

Another problem that is related to $\mbs$ is the Feedback Vertex Set ($\fvs$) problem. The objective of $\fvs$ is the same as that of $\oct$ except the set $S$ has a non-empty intersection with every cycle of the graph (not only the odd ones). The $\fvs$ problem has been extensively studied in graph theory from both hardness~\cite{yannakakis1981node,garey2002computers} and algorithmic~\cite{brandstadt1992improved,daniel1997minimum,kratsch2008feedback,honma2016algorithm} points of view.



We also study a simpler variant of $\mbs$, called the \emph{Maximum Triangle-free Subgraph} ($\tfs$) problem. Let $S$ be a set of $n$ geometric objects in the plane. Then, the objective of the $\tfs$ problem is to compute a maximum-size subset $S'\subseteq S$ such that the intersection graph induced by the objects in $S'$ is triangle free (as opposed to being bipartite).
	

	
	


\paragraph{Related work} The $\mbs$ problem is $\np$-complete for planar graphs with maximum degree four~\cite{DBLP:journals/siamdm/ChoiNR89}. For graphs with maximum degree three, Choi et al.~\cite{DBLP:journals/siamdm/ChoiNR89} showed that for a given constant $ k $ there is a vertex set of size $k$ or less whose removal leaves an  induced bipartite subgraph if and only if there is an edge set of size $k$ or less whose removal leaves a bipartite spanning subgraph. As edge deletion graph bipartization problem is $\np$-complete for cubic graphs~\cite{DBLP:conf/stoc/Yannakakis78}, the $\mbs$ problem is $\np$-complete for cubic graphs. Moreover, the maximum edge deletion graph bipartization problem is solvable in $O(n^3)$ time for planar graphs~\cite{DBLP:journals/siamcomp/Hadlock75,aoshima1977comments}, where $n$ is the number of vertices of the	input graph. Therefore, $\mbs$ is $O(n^3)$-time solvable for planar graphs with maximum degree three. For the vertex-weighted version of the $\mbs$ problem, Baiou et al. \cite{DBLP:journals/siamdm/BaiouB16} showed that the $\mbs$ problem can be solved in $O(n^{3/2}\log n)$ time for planar graphs with maximum degree three. Finally, Cornaz et al.~\cite{DBLP:journals/siamdm/CornazM07} considered the maximum induced bipartite subgraph problem: given a graph with non-negative weights on the edges, the goal is to find a maximum-weight bipartite subgraph. An edge subset $F \subseteq E$ is called independent if the subgraph induced by the edges in $ F $ (incident vertices) is bipartite; otherwise, it is called dependent. They showed that the minimum dependent set problem with non-negative weights can be solved in polynomial time.

The $\oct$ problem is known to be $\np$-complete on planar graphs with degree at most 6~\cite{DBLP:journals/siamdm/ChoiNR89}. For planar graphs with degree at most 3, $\oct$ can be solved in $O(n^3)$ time~\cite{DBLP:journals/siamdm/ChoiNR89} (even the weighted version of the problem). There are several results known concerning the parameterized complexity of $\oct$ (i.e., given a graph $G$ on $n$ vertices and an integer $k$, is there a vertex set $U$ in $G$ of size at most $k$ such that $G\setminus U$ is bipartite). Reed et al.~\cite{DBLP:journals/orl/ReedSV04} first gave an algorithm with running time $O(4^k kmn)$. Lokshtanov et al.~\cite{DBLP:conf/iwoca/LokshtanovSS09} improved this running time to $O(3^k kmn)$. Lokshtanov et al.~\cite{DBLP:conf/fsttcs/LokshtanovSW12} provide an algorithm with running time $O(2^{O(k \log k)} n)$ for planar graphs. Moreover, assuming the exponential time hypothesis, the running time cannot be improved to $2^{O(k)} n^{O(1)}$.

The $\mbs$~problem is also closely related to the Maximum Independent Set ($\mis$) problem. Given a graph $G$, the objective of $\mis$~is to compute a maximum-cardinality subset of vertices such that no two of them are adjacent. Observe that any feasible solution for $\mis$~is also a feasible solution for $\mbs$~and, moreover, a feasible solution $S\subseteq V(G)$ for the $\mbs$~problem provides a feasible solution of size at least $|S|/2$ for the $\mis$~problem. Hence, $\opt(\mis)\leq \opt(\mbs)\leq 2\opt(\mis)$. This implies that an $\alpha$-approximation algorithm for $\mis$ on a class of graphs is a $2\alpha$-approximation for $\mbs$ on the same class. In particular, the $\ptas$es for $\mis$ on unit disks and unit squares~\cite{DBLP:journals/jacm/HochbaumM85} imply polynomial-time $(2+\epsilon)$-approximation algorithms for $\mbs$~on unit disks and unit squares. Moreover, we obtain an $O(\log\log n)$-approximation~\cite{DBLP:conf/soda/ChalermsookC09} (or, $O(\log \opt)$-approximation~\cite{DBLP:conf/wads/BoseCKMMMS19}) algorithm for $\mbs$ on rectangles.

\paragraph{Our results} We first consider the $\mbs$~problem on interval graphs and give a linear-time algorithm for the problem. Moreover, we give an $O(n^2)$-time algorithm that computes a near-optimal solution for the $\mbs$~problem on any circular-arc graph with $n$ vertices (Section~\ref{sec:algo}). Next, we show that the $\mbs$ problem is $\np$-hard on the classes of geometric graphs for which the $\mis$ problem is $\np$-hard (Section~\ref{sec:hard}); this in particular includes unit disks and unit squares. We also extend this result to a corresponding $\w$-hardness result. We obtain a $\ptas$ for the $\mbs$ problem on unit disks and unit squares. For a set of $n$ unit-height rectangles in the plane, we give an $O(n\log n)$-time 2-approximation algorithm for the problem (Section~\ref{sec:appro}).

For a set of $n$ unit disks in the plane, we first give an $O(n^4)$-time algorithm for $\mbs$~problem on $n$ unit disks intersecting a horizontal line where all the centres of the disks lie on one side of the line.  Next we give a 2-approximation algorithm for $\mbs$~problem on $n$ unit disks intersecting a horizontal line where  centre can lie on any side of the line (Section~\ref{subsec:aLine}). For the     $\mbs$~problem on a set of arbitrary unit disks, we give an $O(n^4)$-time 3-approximation algorithm as well as an $O(n\log n)$-time $O(\log n)$-approximation algorithm (Section~\ref{sec3.1}).


Finally, we show that the $\tfs$ problem is $\np$-hard on the intersection graph of axis-parallel rectangles in the plane (Section~\ref{sec:tfs}).

\section{Algorithmic Results}
\label{sec:algo}
In this section, we present our algorithms for the $\mbs$ problem on interval graphs and circular-arc graphs. We start with interval graphs.

\subsection{Interval graphs}
\label{sec:intervals}
In this section, we consider the $\mbs$ problem on a set $S$ of $n$ intervals and give a linear-time algorithm for the problem. Notice that for interval graphs, the $\mbs$ problem is the same as $\fvs$; to the best of our knowledge, the best-known algorithm for solving $\fvs$ on interval graphs takes $O(|V|+|E|))$ time~\cite{lu1997linear}. Since interval graphs are a subclass of chordal graphs, the $\mbs$ problem on interval graphs reduces to the problem of computing a maximum-size subset of intervals in $S$ whose induced graph is triangle free. Consequently, a point can ``stab'' at most two intervals in any feasible solution for the $\mbs$ problem on intervals. Algorithm~\ref{alg:interval} exploits this property to solve the problem exactly.
	
In the following, we assume that (i) the endpoints of intervals in $S$ are $2n$ distinct points on the real line, and (ii) the intervals are sorted from left to right by the increasing order of their right endpoint; we denote them as $I_1,I_2,\dots,I_n$. Moreover, the variable $x$ (resp., $y$) denotes the $x$-coordinate of the rightmost point on the real line such that it is contained in two intervals (resp., one interval) of the current solution computed by the algorithm. For an interval $I$, we denote the left and right endpoints of $I$ by $\lleft(I)$ and $\rright(I)$, respectively.

\begin{algorithm}[t]
	\caption{\textsc{BipartiteInterval($S$)}}
	\label{alg:interval}
	\begin{algorithmic}[1]
		\State let initially $M=\emptyset$;
		\State $x:=-\infty$ and $y:=-\infty$;
		\For{$i:=1$ to $n$}
            \If{$\lleft(I_i)>y$}
                \State $M:=M\cup I_i$ and $y:=\rright(I_i)$;
            \ElsIf{$x<\lleft(I_i)<y$}
                \State $M:=M\cup I_i$, $x:=y$, and $y:=\rright(I_i)$;
            \EndIf
		\EndFor
		\State \Return $M$;
	\end{algorithmic}
\end{algorithm}

\paragraph{Correctness} Let $M_i$, for all $1\leq i\leq n$, denote the set $M$ at the end of iteration $i$ of the for-loop. Consider the following invariant.
\begin{description}
\item[Invariant I.] For all $i=1,\dots,n$, at the end of iteration $i$ of the for-loop, the set $M_i$ is an optimal solution for the set of intervals $\{I_1,I_2,\dots,I_i\}$.
\end{description}

We prove Invariant I by induction on $|S|$. If $|S|=1$, then $M=\{I_1\}$ by line 5 of the algorithm and we are done. Moreover, if $|S|=2$, then there are two cases depending on whether $\{I_1,I_2\}$ form a clique or an independent set. In either case, $M=\{I_1,I_2\}$ and we are done. Now, suppose that Invariant I is true for all $|S|=1,2,\dots,n-1$. Let $S$ be a set of $n$ intervals and consider the set $S\setminus I_n$ (where $I_n$ is the interval with rightmost right endpoint in $S$). By induction hypothesis, let $M_{n-1}$ be the optimal solution for $S\setminus I_n$ computed by the algorithm and consider the values of $x$ and $y$ before returning $M_{n-1}$ in line 13. We must have that either (i) $\lleft(I_n)>y$, (ii) $x<\lleft(I_n)<y$, or (iii) $\lleft(I_n)<x$. In cases (i) and (ii), the algorithm adds $I_n$ to $M_{n-1}$ resulting in an optimal solution. In case (iii), the algorithm returns $M_{n-1}$ without adding $I_n$ to the solution. Observe that this is optimal as no feasible solution can add $I_n$.
	
The algorithm clearly runs in time linear in $n$ and so we have the following theorem.
\begin{theorem}
\label{thm:intervalsMain}
The $\mbs$ problem on a set of $n$ intervals can be solved in $O(n)$ time, assuming that the intervals are already sorted on their right endpoint.
\end{theorem}

\subsection{Circular-arc graphs}
We now give a near-optimal solution for the $\mbs$ problem on circular-arc graphs. For an optimization problem, a \emph{near-optimal} solution is a feasible solution whose objective function value is within a specified range from the optimal objective function value. A circular-arc graph is the intersection graph of arcs on a circle. That is, every vertex is represented by an arc, and there is an edge between two vertices if and only if the corresponding arcs intersect. Observe that interval graphs are a proper subclass of circular-arc graphs. For the rest of this section, let $G=(V,E)$ be a circular-arc graph and assume that a geometric representation of $G$ (i.e., a set of $|V(G)|$ arcs on a circle ${\cal C}$) is given as part of the input. First, we prove the following lemmas.
\begin{lemma}
\label{lem:onecycle}
If $G$ is triangle-free, then it can have at most one cycle.
\end{lemma}
\begin{proof}
Suppose for the sake of contradiction that $G$ has more than one cycle. Let $A_1$ and $A_2$ be two cycles of $G$. Now, since $G$ is a triangle-free circular-arc graph, the corresponding arcs of the vertices of any cycle in $G$ together cover the circle $\mathcal{C}$. So, there must exist three distinct vertices $v\in A_1, u\in A_1$ and $w\in A_2$ such that $v, u, w$ are pairwise adjacent. Which is a contradiction to the fact that $G$ is triangle-free.
\end{proof}
	
\begin{lemma}
\label{lem:fvs_circular}
If $B$ and $T$ are optimal solutions for the $\mbs$ and $\tfs$ problems on $G$, respectively, then $|T|-1\leq |B|\leq|T|$.
\end{lemma}
\begin{proof}
Since a bipartite subgraph contains no triangle, $|B|\leq|T|$. Now, if $G[T]$ (i.e., the subgraph of $G$ induced by $T$) is odd-cycle free, then it induces a bipartite subgraph. Otherwise, $G[T]$ can have at most one cycle by Lemma~\ref{lem:onecycle}. If this cycle is odd, then by removing any single vertex form the cycle, we obtain a bipartite subgraph of $G$ with size at least $|T|-1$.
\end{proof}
	
Since $G[T]$ contains at most one cycle, following lemma trivially holds.
\begin{lemma}
\label{lem:forest}
If $H$ is a maximum-size induced forest in $G$, then $|V(H)|\geq |T|-1$.
\end{lemma}
	
By the above lemmas, our goal now is to find a maximum acyclic subgraph $H$ of $G$. Notice that there must be a clique $K$ ($|K|\ge 1$) in $G$ that is not in $H$. Now, for each arc $u$ in the circular-arc representation of $G$, let $l(u)$ and $r(u)$ denote the two endpoints of $u$ in the clockwise order of the endpoints $u$. Then, we consider two vertex sets $S_{u}^1= \{w \colon w \in V,  l(u) \notin [l(w),r(w)] \}$ and $S_{u}^2= \{z \colon z \in V,  r(u) \notin [l(z),r(z)] \}$. Both $S_{u}^1$ and $S_{u}^2$ are interval graphs. Since there are $n$ vertices in $G$, we compute $2n$ interval graphs in total. Then, for each of these interval graphs, we apply Algorithm~\ref{alg:interval} to compute an optimal solution for $\mbs$, and will return the one with maximum size as the final solution. Since Algorithm~\ref{alg:interval} runs in $O(n)$ time, the total time to find $H$ is $O(n^2)$; so we have the following theorem.
\begin{theorem}
Let $OPT$ be a maximum-size induced bipartite subgraph of a circular-arc graph $G$ with $n$ vertices. Then, there is an algorithm that computes an induced bipartite subgraph $H$ of $G$ such that $|V(H)|\geq |OPT|-1$. The algorithm runs in $O(n^2)$ time.
\end{theorem}

\section{$\np$-Hardness}
\label{sec:hard}
In this section, we show that the $\mbs$ problem is $\np$-complete on the classes of geometric intersection graphs for which $\mis$ is $\np$-complete. The $\mis$ problem is known to be $\np$-complete on a wide range of geometric intersection graphs, even restricted to unit disks and unit squares~\cite{clark1990unit}, 1-string graphs~\cite{kratochvil1990independent}, and B$ _1 $-VPG graphs~\cite{lahiri2015maximum}. Let $G=(V,E)$ be an intersection graph induced by a set $S$ of $n$ geometric objects in the plane. We construct a new graph $G'$ from the disjoint union of two copies of $G$ by adding edges as follows. For each vertex in $V$, we add an edge from each vertex in one copy of $G$ to the corresponding vertex in the other copy. For each edge $(u,v)\in E$, we add four edges $(u,v), (u',v'),(u,v'),$ and $(v,u')$ to $G'$, where $u'$ and $v'$ are the corresponding vertices of $u$ and $v$, respectively in the other copy. Graph $G'$ is the intersection graph of $2n$ geometric objects $S$, where each object has occurred twice in the same position. Clearly, the number of vertices and edges in $G'$ are polynomial in the number of vertices of $G$; hence, the construction can be done in polynomial time.
\begin{lemma}
\label{lem:isAndMBS}
$G$ has an independent set of size at least $k$ if and only if $G'$ has a bipartite subgraph of size at least $2k$.
\end{lemma}
\begin{proof}
Let $U$ be an independent set of $G$ with $|U|\geq k$. Let $H$ be the subgraph of $G'$ induced by $U$ along with all the corresponding vertices of $U$ in the other copy. Then, $H$ is a bipartite subgraph with size at least $2k$. Conversely, if $G'$ has a bipartite subgraph of size at least $2k$, then $G'$ must have an independent set of size at least $k$. By the construction of $G'$, if $G'$ has an independent set of size at least $k$, then $G$ must have an independent set of size at least $k$.\hfill
\end{proof}

By Lemma~\ref{lem:isAndMBS}, we have the following theorem.
\begin{theorem}
\label{thm-mbs-np-hard}
The	$\mbs$~problem is $\np$-complete on the classes of geometric intersection graphs for which $\mis$ is $\np$-complete.
\end{theorem}

\paragraph{Remark} By the definition of parameterized reduction~\cite{DBLP:series/mcs/DowneyF99}, one can verify that the above reduction is in fact a parameterized reduction and so we have the following result.
	
	
	
	

\begin{corollary}
\label{thm-mbs-w-hard} 
The	$\mbs$ problem is $\w$-complete on the classes of geometric intersection graphs for which  $\mis$ is $\w$-complete.
\end{corollary}

Marx~\cite{DBLP:conf/esa/Marx05,marx2006parameterized} proved that $\mis$ is $\w$-complete on unit squares, unit disks, and even unit line segments. As such, by Corollary~\ref{thm-mbs-w-hard}, the $\mbs$ problem is $\w$-complete on all these geometric intersection graphs.

\section{Approximation Algorithms}
\label{sec:appro}
Recall that since $\mis$ is $\np$-complete on unit disks and unit squares, the $\mbs$ problem is $\np$-complete on these graphs by Theorem~\ref{thm-mbs-np-hard}. In this section, we first give $\ptas$es for $\mbs$ on both unit squares and unit disks, and will then consider the problem on unit-height rectangles.

\subsection{Unit disks and unit squares}
\label{sec:unit}
We first show the $\ptas$ for unit disks, and will then discuss it for unit squares as well as the \emph{weighted} $\mbs$ problem.
	
Let $S$ be a set of $n$ unit disks in the plane, and let $k>1$ be a fixed integer. A $\ptas$ running in $O(n^{O(1)}\cdot n^{O(1/\epsilon^2)})$ time, for any $\epsilon>0$, is straightforward using the shifting technique of Hochbaum and Maass~\cite{DBLP:journals/jacm/HochbaumM85} and the following packing argument: for an instance of the $\mbs$ problem, where the unit disks lie inside a $k\times k$ square, an optimal solution cannot have more than $k^2$ unit disks. Hence, we can compute an exact solution for such an instance of the problem in $O(n^{O(1)}\cdot n^{O(k^2)})$ time. Consequently, by setting $k=1/\epsilon$, we obtain a $\ptas$ running in time $O(n^{O(1)}\cdot n^{O(1/\epsilon^2)})$.

To improve the running time to $O(n^{O(1)}\cdot n^{O(1/\epsilon)})$, we rely on the shifting technique again, but instead of applying the plane partitioning twice, we only partition the plane into horizontal slabs and solve the $\mbs$ problem for each of them exactly. This gives us the desired running time for our $\ptas$. We next describe the details of how to solve $\mbs$ exactly for a slab.
	
\paragraph{Algorithm for a slab} Let $H$ be a horizontal slab of height $k$ and let $D\subseteq S$ be the set of disks that lie entirely inside $H$. The idea is to build a vertex-weighted directed acyclic graph $G$ such that finding a maximum-weight path from the source vertex to the target vertex corresponds to an exact solution for the $\mbs$ problem~\cite{DBLP:conf/jcdcg/Matsui98}. To this end, let $a$ and $b$ ($a<b$) be two integers such that every disk in $D$ lies inside the rectangle $R$ bounded by $H$ and the vertical lines $x=a$ and $x=b$. Partition $R$ vertically into unit-width boxes $B_i$, where the left side of $B_i$ has the $x$-coordinate $a+i$, for all integers $0\leq i<b-a$; let $D_i\subseteq D$ denote the set of disks whose centers lie inside $B_i$. Since $B_i$ has height $k$ and width 1, we can compute all feasible (not necessarily exact) solutions for the $\mbs$ problem on $D_i$ in $O(n^{O(1)}\cdot n^{O(k)})$ time; let ${\cal M}_i$ be the set of all such feasible solutions. We now build a directed vertex-weighted acyclic graph $G$ as follows. The vertex set of $V(G)$ is $V\cup \{s,t\}$, where $V$ has one vertex for each solution in ${\cal M}_i$, for all $i$. Moreover, the weight of each vertex is the number of disks in the corresponding bipartite graph. For every pair $i,j$, where $1\leq i<j<n$, consider two solutions $M\in {\cal M}_i$ and $M'\in {\cal M}_j$. Then, there exists an edge from the vertex of $M$ to that of $M'$ in $G$ if the intersection graph induced by the disks in $M\cup M'$ is bipartite. Finally, for all $i$ and for all $M\in {\cal M}_i$: there exists an edge from $s$ to $M$, and there exists an edge from $M$ to $t$. The weights of vertices $s$ and $t$ are zero.
\begin{lemma}
\label{lem:bipartiteGraph}
The $\mbs$ problem has a feasible solution of size $k$ on $G$ if and only if there exists a directed path from $s$ to $t$ with the total weight $k$.
\end{lemma}
\begin{proof}
For a given directed $st$-path with total weight $k$, let $X$ be the union of all the disks corresponding to the internal vertices of this path. Then, the intersection graph of $X$ is bipartite because the disks in $X\cap {\cal M}_i$ are disjoint from the disks in $S\cap {\cal M}_j$ when $j>i+1$. Moreover, when $j=i+1$, the disks in $X\cap ({\cal M}_i\cup {\cal M}_j)$ must form an induced bipartite graph by the definition of an edge in $G$. Since the total weight of the vertices on the path is $k$, we have $|X|=k$. On the other hand, let $Y$ be a feasible solution of size $k$ for the $\mbs$ problem on $G$. Then, the intersection graph of disks in $Y\cap D_i$ is bipartite, for all $i$. Hence, selecting the vertices corresponding to $Y\cap D_i$ for all $i$ gives us a path with total weight $k$ from $s$ to $t$.\hfill
\end{proof}
	
By Lemma~\ref{lem:bipartiteGraph}, the $\mbs$ problem for $H$ reduces to the problem of finding the maximum-weighted path from $s$ to $t$ on $G$. The number of vertices of $G$ that correspond to feasible solutions for the $\mbs$ problem on disks in $S\cap D_i$ is bounded by $O(n^{O(k)})$, which is the bound on the number of vertices of $G$ that correspond to these feasible solutions. Hence, we can compute the edge set of $G$ in $O(n^{O(1)}\cdot n^{O(k)})$ time (we can check whether a subset of disks form a bipartite graph in $O(n^{O(1)})$ time). Since $G$ is directed and acyclic, the maximum-weighted $st$-path problem can be solved in linear time. Therefore, by setting $k=1/\epsilon$, we have the following theorem.
\begin{theorem}
\label{thm:ptasDisks}
There exists a $\ptas$ for $\mbs$ on unit disks that runs in $O(n^{O(1)}\cdot n^{O(1/\epsilon)})$ time, for any $\epsilon>0$.
\end{theorem}
	
\paragraph{$\ptas$ on unit squares} One can verify that the above algorithm can be applied to obtain a $\ptas$ for $\mbs$ on a set of $n$ unit squares, as well. Moreover, the algorithm extends to the weighted $\mbs$ problem on unit disks and unit squares. The only modification is, instead of assigning the number of disks (resp., squares) in a solution as the weight of the corresponding vertex, we assign the total weight of the disks (resp., squares) in the solution as the vertex weight.
\begin{theorem}
There exists a $\ptas$ for the $\mbs$ problem on unit squares running in $O(n^{O(1)}\cdot n^{O(1/\epsilon)})$ time, for any $\epsilon>0$. Moreover, the weighted $\mbs$ problem also admits a $\ptas$ running within the same time bound on unit disks and unit squares.
\end{theorem}
	
\paragraph{A 4-approximation on unit disks} Recall from Section~\ref{sec:introduction} that $\opt(\mis)\leq \opt(\mbs)\leq 2\opt(\mis)$. Nandy et al.~\cite{DBLP:journals/ipl/NandyPR17} designed a factor-2 approximation algorithm for the $\mis$ problem on unit disks, which runs in $O(n^2)$ time. Consequently, we obtain an $O(n^2)$-time 4-approximation algorithm for the $\mbs$~problem on unit disks.

\subsection{Unit-height rectangles}
\label{sec:rectangles}
Here, we give an $O(n\log n)$-time 2-approximation algorithm for $\mbs$ on a set of $n$ unit-height rectangles. To this end, suppose that the bottom side of the bottommost rectangle has $y$-coordinate $a$ and the top side of the topmost rectangle has $y$-coordinate $b$. Consider the set of horizontal lines $y:=a+i+\epsilon$ for all $i=0,\dots,b$, where $\epsilon>0$ is a small constant; we may assume w.l.o.g. that each rectangle intersects exactly one line. Ordering the lines from bottom to top, let $S_i$ be the set of rectangles that intersect the horizontal line $i$. We now run \textsc{BipartiteInterval}$(S)$, once for when $S=S_1\cup S_3\cup S_5\dots$ and once for when $S=S_2\cup S_4\cup S_6\dots$, and will then return the largest of these two solutions. We perform an initial sorting that takes $O(n\log n)$ time, and \textsc{BipartiteInterval}$(S)$ runs in $O(n)$ time. This gives us the following theorem.
\begin{theorem}
There exists an $O(n\log n)$ time 2-approximation algorithm for the $\mbs$ problem on a set of $n$ unit-height rectangles in the plane.
\end{theorem}

\section{$\mbs$~on Unit Disk Graphs}

In this section, we look the $\mbs$ problem on unit disk graphs. We first consider the problem when the input disks intersect a horizontal line and all centres lie on one side of the line. Next we  give a 2-approximation algorithm for the $\mbs$ problem on the unit disks intersecting a horizontal line (i.e., without the restriction of side of the line their centres lie). Our 2-approximation algorithm is based on (i) partitioning the input disks into two sets, depending on which side of the line their centres lie, and (ii) solving the $\mis$ problems independently in these two sets. Then we will discuss how to obtain the factor-2 approximation. Our results for the $\mbs$ problem on arbitrary unit disk graphs (i.e., without the restriction of intersecting a line) will then follow.

\subsection{Unit disks intersecting a line}
\label{subsec:aLine}
Here, we are given $n$ unit disks $\mathcal{D}=\{D_1, D_2, \dots, D_n \}$ intersecting a straight line $L$, where the centre of every disk in $\mathcal{D}$ lies on one side of $L$. We assume w.l.o.g. that all the disks intersect the $x$-axis and all the centers have non-negative $y$-coordinates. Let $\mathcal{D}=\{D_1, D_2, \dots, D_n \}$ be a set of $n$ unit disks in the plane that are intersected by the $X$-axis and all the centers of the disks have non-negative $y$-coordinate. We give an $O{(n^4)}$-time algorithm to solve the $\mbs$ problem on $G_{\mathcal{D}}$; we assume that the optimal solution has size greater than 2 (if it is at most 2, then we can find it easily).

We first give some notation. Let $G_{\mathcal{D}}$ denote the intersection graph of disks in $\mathcal{D}$. Each disk corresponds to a vertex in $G_{\mathcal{D}}$ and there is an edge between vertices if the corresponding disks intersect. In the rest of this section, we use ``disks'' and ``vertices'' in $G_{\mathcal{D}}$ interchangeably. We also use $x(p)$ and $y(p)$ to denote, respectively, the $x$- and $y$-coordinates of a point $p$. Let $d(p,q)$ denote the Euclidean distance between two points $p$ and $q$. Let $c_i$ denote the centre of the disk $D_i$ for all $1\leq i\leq n$. W.l.o.g., we assume that $x(c_1)\leq x(c_2)\leq\dots\leq x(c_n)$. By $\rectangle(p_1, p_2, p_3, p_4)$, we mean the rectangle with corner points $p_1, p_2, p_3, p_4$ in clockwise order.

We first show that the $\mbs$ problem and the $\tfs$ problem are equivalent on $G_\mathcal{D}$. To this end, we need the following lemmas.
\begin{lemma}
\label{lem-ind}
Let $D_i, D_j, D_k \in \mathcal{D}$ be three disks with centers $c_i , c_j$ and $c_k$, respectively, where $x( c_i ) \leq x ( c_j ) \leq x ( c_k )$. If $ D_i \cap D_j = \emptyset$ and $ D_j \cap D_k =\emptyset $, then $ D_i\cap D_k= \emptyset $.
\end{lemma}
\begin{proof}
We prove this by contradiction. Suppose that $D_i\cap D_k\neq\emptyset$. Let $ R $ be the rectangle such that  $R=\rectangle((x(c_i),0), (x(c_i),1), (x(c_k),1),(x(c_k),0))$. Since $x( c_i ) \leq x ( c_j ) \leq x ( c_k )$, $c_j$ must belong to $ R $. Let $ m=(x(c_i)+x(c_k))/2$. We partition the rectangle $R$ into two rectangles $ A $ and $ B $ in the following way:
\[
A=\rectangle((x(c_i),0),(x(c_i),1),(m,1),(m,0)),
\]
\[
B=\rectangle((m,0), (m,1), (x(c_k),1),(x(c_k),0)).
\]
Since $D_i\cap D_k \neq\emptyset$, we have $d(c_i,c_k)\leq 2$. Then, $d((x(c_i),0), (m,0))\leq 1$ and $d((m,0),(x(c_k),0))\leq 1 $. Now, if $c_j\in A$ then $d(c_i,c_j) \leq \sqrt{2}$ and so $ D_i \cap D_j \neq \emptyset $. Otherwise, if $c_j\in B$, then $ d(c_j,c_k) \leq \sqrt{2} $ and so $D_j\cap D_k\neq \emptyset$. See Figure~\ref{fig-lem1} for an illustration. Since both cases lead to a contradiction, the lemma holds.
\end{proof}

\begin{figure}[t]
\centering
\includegraphics[width=2in,height=1in]{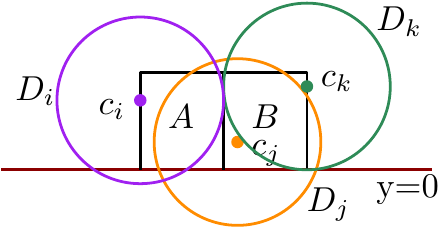}
\caption{An illustration in supporting the proof of Lemma~\ref{lem-ind}.} 
\label{fig-lem1}
\end{figure}

It is easy to observe that  $G_{\mathcal{D}}$ can contain an induced 4-cycle; see Figure~\ref{fig-4cycle} for an example.

\begin{figure}[ht!]
	\centering
	\includegraphics[width=2.1in,height=.9in]{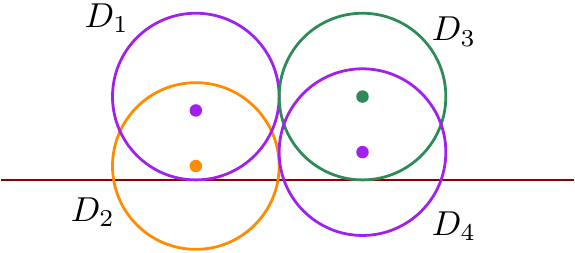}
	\caption{An example of a 4-cycle.} 
	\label{fig-4cycle}
\end{figure}

\begin{lemma}
\label{lem-five}
There is no induced cycle of length at least $5$ in $G_{\mathcal{D}}$.
\end{lemma}
\begin{proof}
	 We prove this lemma by contradiction. Let $G_{\mathcal{D}}$ contains an induced cycle $ C = \langle D_{1}, D_2, \dots, D_r, \dots, D_k \rangle $ of length $ k \geq 5 $. W.l.o.g. we assume that $  \forall i,~ 2\leq i \leq k,~ x(c_1) \leq x(c_i)$.  Let $ D_r $ be the right-most disk according to the increasing order of $ x$-coordinates of the centers of disks in the cycle $ C $. Now there are two disjoint paths $ \pi_1 $ and $ \pi_2 $ in $ C $ between the vertices corresponding to $ D_1 $ and $ D_r $. Further we assume  that the number of disks in $ \pi_1 $ is atmost the number of disks in $ \pi_2 $. The other case is analogous. Our proof is based on the number of vertices in $ \pi_1 $. 
	 
	 First consider the case that $ D_1 \cap D_r \neq \phi $. In this case $ \pi_1 $ consists of disks $ \{D_1, D_2 (=D_r) \}$ and $ \pi_2 $ consists of disks $ \{D_2, D_3, \dots, D_k, D_1\} $. Note that the number of disks in $ \pi_2 $ is atleast 5. So there must be a disk $ D_i \in \pi_2$ such that $ D_i \cap D_1 = \phi$ and $ D_i \cap D_2= \phi$. As $ x (c_1) \leq x(c_i) \leq x(c_2) $, by Lemma \ref{lem-ind},  $ D_1 \cap D_2 = \phi$. This contradicts our assumption. 
	 
	 Now we consider the case that $ \pi_1= (D_1, D_2, \dots, D_r ) $, where $ r \geq 3 $ and $ \pi_2= (D_1, D_k, \dots, D_{r+1}, D_r) $. Note that the number of disks in $ \pi_1 $ is atleast three and the number of disks in $ \pi_2 $ is atleast four. Now we consider following three subcases based on location of centers of $ D_2,\dots, D_{r-1} $.
	 	
	 	\noindent{\bf Case 1:}  $x(c_{k}) \leq x(c_j) \leq  x(c_{k-1})$ holds for atleast one $ j $ $\{ 2, \dots, r-1\} $.  As $ D_k \cap D_j = \emptyset$ and $ D_j \cap D_{k-1} = \emptyset $, so  by Lemma \ref{lem-ind}, $ D_k $ and $ D_{k-1} $ should not intersect each other. This  contradicts our assumption that $ D_k \cap D_{k-1} \neq \phi$. 
	 	
	 	\noindent{\bf Case 2:  $x(c_{i}) \leq x(c_k),~ \forall i, ~2 \leq i \leq r-1 $}.  As $ D_{r-1} \cap D_{k} = \emptyset$ and $ D_{k} \cap D_r = \emptyset $, so  by Lemma \ref{lem-ind}, $ D_{r-1} $ and $ D_r $ should not intersect each other. This  contradicts our assumption that $ D_{r-1} \cap D_{r} \neq \phi$. 
	 	
	 	\noindent{\bf Case 3:  $x(c_{k-1}) \leq x(c_{i}) ,~ \forall i,~ 2 \leq i \leq r-1 $}.  As $ D_1 \cap D_{k-1} = \emptyset$ and $ D_{k-1} \cap D_2 = \emptyset $, so  by Lemma \ref{lem-ind}, $ D_1 $ and $ D_2 $ should not intersect each other. This  contradicts our assumption that $ D_1 \cap D_{2} \neq \phi$. 
	 \end{proof}

\noindent By Lemma~\ref{lem-five}, the $\mbs$ and $\tfs$ problems are equivalent on $G_{\mathcal{D}}$. Therefore, in the following, we focus on solving the $\tfs$~problem on $G_{\mathcal{D}}$.

\paragraph{A dynamic-programming algorithm} Let ${\cal D}[i] = \{D_i,D_{i+1},\dots,D_n\}$. For each triple $(i,j,k)$, where $1 \leq i<j<k\leq n$, we define $B[i,j,k]$ to be the size of a maximum induced triangle-free subgraph containing $D_i, D_j$ and $D_k \in D_i\cup D_j\cup {\cal D}[k]$. Let $D[i,j,k]=\{D_l \colon D_l\in {\cal D}[k+1] \mbox{ and } \{D_i, D_j, D_k, D_{\ell}\} \mbox{ induces a } K_3 \mbox{-free }$ $\mbox{subgraph}\}$. We now describe how to compute $ B[i,j,k]$ for each triple $(i,j,k)$. If $D_i, D_j$ and $D_k$ form a $K_3$, then clearly $B[i,j,k]=0$ as there is no bipartite subgraph containing $D_i, D_j$ and $D_k$. If $D_i, D_j$ and $D_k$ do not form a $K_3$ and $D[i,j,k] =\emptyset$, then $B[i,j,k]=3$ and in that case $D_i, D_j$ and $D_k$ are the disks that induce maximum induced triangle-free subgraph in $D_i \cup D_j \cup {\cal D}[k]$. On the other hand, if $D_i, D_j, D_k$ do not form a $ K_3 $ and $D[i,j,k]\neq\emptyset$, then $B[i,j,k]$ is one more than the size of a maximum induced triangle-free subgraph containing $ D_j, D_k, D_l $ among all $ D_l \in D[i,j,k]$. Hence, we obtain the following recurrence for $B[i,j,k]$.
\begin{equation*}
{\small B[i,j,k]=\begin{cases}
0, & \text{$D_i, D_j, D_k$ form a $ K_3 $},\\
3, & \text{$D_i, D_j, D_k$ do not form a $ K_3 $ and $D[i,j,k] =\emptyset$},\\
1+ \underset{ \forall D_l \in D[i,j,k]}{\mathrm{max}}
B[j,k,l],  \ & \text{otherwise.}
\end{cases}
}
\end{equation*}

The size of an optimal solution is the maximum value in the table $B$. To show the correctness of our algorithm, we prove the following.
\begin{lemma}
\label{lem:DP}
Let $ B[i,j,k] >3, $ and $D_{\ell} \in D[i,j,k] $. Then, no two disks corresponding to $ B[j,k,{\ell}] $ form a triangle with $ D_i $.
\end{lemma}

We need several helper lemmas before proving Lemma~\ref{lem:DP}.
\begin{lemma}
\label{lem-three}
Let $D_i, D_j, D_k, D_{\ell}\in \mathcal{D}$ be four disks with $x( c_i ) \leq x ( c_j ) \leq x ( c_k) \leq x(c_{\ell})$. If $D_i \cap D_{\ell} \neq \emptyset$ then the subgraph induced by $\{D_i, D_j, D_k,D_{\ell}\}$ is $K_{1,3}$-free.
\end{lemma}
\begin{proof}
We prove the lemma by contradiction. Suppose that the subgraph induced by $\{D_i, D_j, D_k, D_{\ell}\}$ is isomorphic to $K_{1,3}$. As $D_i \cap D_{\ell} \neq \emptyset$, either (i) $D_i$ intersects both $ D_j $ and $ D_k $, where $D_j, D_k, D_{\ell} $ are pairwise independent, or (ii)  $D_{\ell}$ intersects both $ D_j $ and $ D_k $, where $D_i, D_j, D_{k} $ are pairwise independent,. We prove the lemma for (i); the proof for (ii) is similar. Consider the rectangle $R=\rectangle((x(c_i),0), (x(c_i),1), (x(c_{\ell}),1),(x(c_{\ell}),0)))$. Since $x( c_i ) \leq x ( c_j ) \leq x ( c_k) \leq x(c_{\ell})$, we have that $c_j$ and $c_k$ must belong to $R$. Let $ m=(x(c_i)+x(c_{\ell}))/2$. We partition the rectangle $R$ into two rectangles $A$ and $B$ in the following way:
\[
A =\rectangle((x(c_i),0), (x(c_i),1), (m,1), (m, 0)),
\]
\[
B= \rectangle((m,0), (m,1), (x(c_{\ell}),1), (x(c_{\ell}),0)).
\]

Since $ D_i \cap D_{\ell} \neq \emptyset$, we have $d(c_i,c_{\ell})\leq 2$. Therefore, $d((x(c_i),0), (m,0))\leq 1$ and $d((m,0),(x(c_{\ell}),0))\leq 1$. If $c_j, c_k\in A$ or $c_j, c_k\in B$, then $d(c_j,c_k)\leq\sqrt{2}<2 $ and so $D_j\cap D_k \neq \emptyset$. This contradicts that $D_j, D_k, D_{\ell} $ are pairwise independent. Otherwise, $c_j\in A$ and $c_k \in B$. Since $c_{\ell}\in B$ so $ d(c_k,c_{\ell}) \leq \sqrt{2}<2 $. Therefore, $D_k\cap D_{\ell}\neq \emptyset$, this contradicts that $D_j, D_k, D_{\ell} $ are pairwise independent.
\end{proof}

\begin{figure}[t]
\centering
\includegraphics[width=2.1in,height=1in]{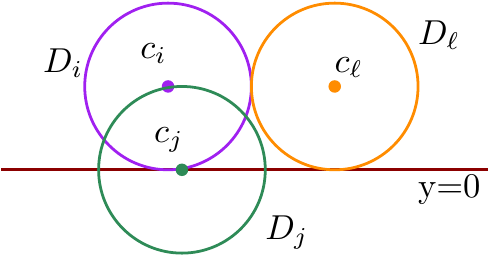}
\caption{Existence of $K_{1,2}$ in one side.} 
\label{fig-lem3}
\end{figure}

\begin{figure}[ht!]
\centering
\includegraphics[width=2.9in,height=1.3in]{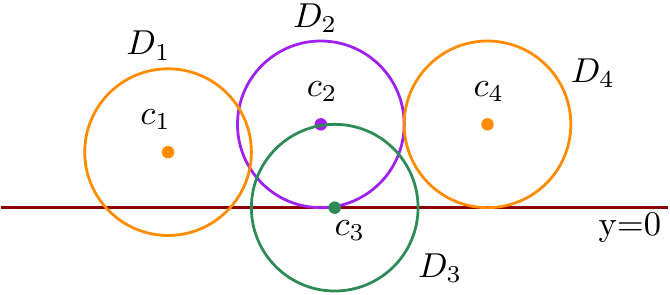}
\caption{Existence of $K_{1,3}$ in $ G_{\mathcal{D}}$.} 
\label{fig-lem4}
\end{figure}

\begin{lemma}
\label{lem-four}
The graph $G_{\mathcal{D}}$ is $K_{1,4}$-free.
\end{lemma}
\begin{proof}
We prove the lemma again by contradiction. Suppose  there is an induced subgraph $K_{1,4}$ with vertex set $ \{D_i\colon 1\leq i\leq 5\}$. Assume w.l.o.g. that $x( c_1 ) \leq x ( c_2 ) \leq x ( c_3) \leq x(c_4) \leq x(c_5)$. In $ \{D_i\colon 1\leq i\leq 5\}$, there must be four pair-wise independent disks and one disk (say, $D$) that intersects each of that four disks. If $D=D_1~ \text{or}~D_4$. Then  the subgraph induced by $ \{D_i\colon 1\leq i\leq 4\}$ is isomorphic to $K_{1,3}$. This contradicts  Lemma~\ref{lem-three}. Similarly if $D=D_2 ~\text{or}~D_5$. Then  the subgraph induced by $ \{D_j\colon 2\leq j\leq 5\}$ is isomorphic to $K_{1,3}$. This contradicts   Lemma~\ref{lem-three}. Hence only possibility is that $D=D_3$ and the edge set in $K_{1,4}$ is $\{(D_1, D_3),(D_2, D_3),(D_3, D_4),(D_3, D_5)\}$.

Now, consider the rectangle $R=\rectangle((x(c_3)-2,0),(x(c_3)-2,1),(x(c_3)+2,1),(x(c_3)+2,0)))$. We partition $R$ into three rectangles $A, B$ and $C$ as follows:
\[
A =\rectangle((x(c_3)-2,0), (x(c_3)-2,1), (x(c_3)-0.5,1), (x(c_3)-0.5,0)),
\]
\[
\hspace{5mm} B =\rectangle((x(c_3)-0.5,0), (x(c_3)-0.5,1), (x(c_3)+0.5,1), (x(c_3)+0.5,0)),
\]
\[
C =\rectangle((x(c_3)+0.5,0), (x(c_3)+0.5,1), (x(c_3)+2,1), (x(c_3)+2,0)).
\]
Since $c_i$ belongs to $R$, for $1\leq i\leq 5$, by the pigeonhole principle, one of $A$, $B$ and $C$ must contain two of $\{c_1, c_2, c_4, c_5\}$. The distance between any two points in the rectangles $A$, $B$ or $C$ is $\leq 2$. Therefore, $\{D_1, D_2, D_4, D_5\}$ are not pairwise independent. Hence, $G_{\mathcal{D}} $ is $K_{1,4}$-free.	
\end{proof}

\begin{lemma}
\label{lem-3ind}
Let $D_1, D_2, D_3, D_4, D_5$ be five disks in  $ \mathcal{D} $ with $x( c_i ) \leq x ( c_{i+1} )$ where $ 1 \leq i \leq 4$ . If $D_1\cap D_5 \neq\emptyset$ then the disks $D_2, D_3, D_4$ do not form an independent set.
\end{lemma}
\begin{proof}
The proof proceeds analogous to the proof of Lemma~\ref{lem-four}. Since $D_1$ and $D_5$ intersect, $d(c_1,c_5)\leq 2$, where $c_1$ and $c_5$ are the centres of the disks $D_1$ and $D_5$, respectively. Now, consider the rectangle
\[
\hspace{13mm} R=\rectangle((x(c_1),0), (x(c_1),1), (x(c_5),1),(x(c_5),0)),
\]
and let $m=(x(c_1)+x(c_5))/2$. We divide $R$ into two rectangles $A$ and $B$, where
\[
A =\rectangle((x(c_1),0),(x(c_1),1),(m,1),(m,0))
\]
\[
B=\rectangle((m,0),(m,1),(x(c_5),1),(x(c_5),0)).
\]

Now, since each of $\{c_i\colon 1\leq i\leq 5\}$ belongs to $R$, one of $A$ and $B$ must contain two of $\{c_2, c_3, c_4\}$. The distance between any two points in the rectangle $A$ or $B$ is at most two. Therefore $D_2, D_3, D_4$ do not form an independent set. 
\end{proof}

We are now ready to prove Lemma~\ref{lem:DP}.
\begin{proof}[Proof of Lemma~\ref{lem:DP}]
	
By assumption $ B[i,j,k] >3 $ and $ D_{\ell} \in \{D_{k+1}, \dots, D_n\} $ where 	the subgraph induced by $\{ D_i, D_j, D_k, D_{\ell}\}  $ is triangle-free. Let $ \Psi_{j,k,\ell} $ be the set of disks that defines the count for $ B[j,k,\ell] $. Since $ B[i,j,k] >3, ~\Psi_{j,k,\ell} \neq \phi $. We need to show that for any pair of disks $ D, D' \in \Psi_{j,k,\ell} $, the subgraph induced by $ \{D_i, D, D'\} $  is triangle-free. We prove this by contradiction. Suppose for a contradiction that $D$ and $D'$ be two disks corresponding to $B[j,k,\ell]$ such that 	$D_i,D$ and $D'$ form a triangle. As the subgraph induced by $\{ D_j, D_k, D_{\ell},D, D'\}  $ is triangle-free. By the definition of $D[i,j,k]$, at most one of $D$ and $D'$ belongs to $\{D_j, D_k, D_{\ell}\}$.

\noindent{\bf Case 1: $\{D, D'\}\cap\{D_j, D_k, D_{\ell}\}\neq\emptyset$}. Assume w.l.o.g. that $D'\in \{D_j, D_k, D_{\ell}\}$. We consider several cases.
\begin{itemize}
\item $D'=D_j$. Since $D_i\cap D\neq \emptyset$, by Lemma~\ref{lem-ind}, (i) either $D_i \cap D_k\neq\emptyset$ or $D_k\cap D\neq\emptyset$ as well as (ii) either $D_i\cap D_{\ell}\neq\emptyset$ or $D_{\ell}\cap D\neq\emptyset$. Similarly, since $D_j\cap D\neq\emptyset$, by Lemma~\ref{lem-ind}, (i) either $D_j \cap D_{\ell}\neq\emptyset$ or $D_{\ell} \cap D\neq\emptyset$. Suppose that $D_k \cap D\neq\emptyset$. Now, if $D_{\ell} \cap D\neq\emptyset$, then $\{D_j,D_k,D_l,D\}$ forms a $K_{1,3}$. This  contradicts Lemma~\ref{lem-three}. Therefore, $D_{\ell}\cap D=\emptyset$; consequently, $D_j\cap D_{\ell}\neq\emptyset$ and $D_i\cap D_{\ell}\neq\emptyset$. Then, the disks $D_i, D_j, D_{\ell}$ form a triangle. A contradiction. So, $D_k$ and $D$ should not intersect each other. Then, $D_i\cap D_k\neq\emptyset$ and $D_j\cap D_k\neq\emptyset$. Consequently, the disks $ D_i, D_j, D_k $ form a triangle and so $B[i,j,k]=0$. This contradicts our assumption that $B[i,j,k]>3$.

\item $D'=D_k$.   Since $D_i\cap D\neq\emptyset$, by Lemma~\ref{lem-ind}, either $D_i\cap D_j\neq\emptyset$ or $D_j\cap D\neq\emptyset$. Suppose that $D_j\cap D\neq\emptyset$. Since $D_k\cap D\neq\emptyset$, either $D_k\cap D_{\ell}\neq\emptyset$ or $D_{\ell}\cap D\neq\emptyset$ (again by Lemma~\ref{lem-ind}). Let $D_{\ell}\cap D\neq \emptyset$. Then, $\{D_j, D_k, D_{\ell}, D\}$ forms a $K_{1,3}$. This  contradicts Lemma~\ref{lem-three}. So, $D_{\ell}\cap D=\emptyset $, that means that $D_k \cap D_{\ell}\neq\emptyset$ and $D_i \cap D_{\ell}\neq \emptyset$. Consequently, the disks $D_i, D_k, D_{\ell}$ form a triangle. A contradiction. This means that the disks $D_j$ and $D$ do not intersect each other and so $D_i \cap D_j\neq\emptyset$. Now, if $D_i\cap D_{\ell}\neq\emptyset$, then $\{D_i, D_j, D_k, D _{\ell}\}$ forms a $K_{1,3}$. This  contradicts Lemma~\ref{lem-three}. Thus, the disks $D_i$ and $D_{\ell}$ do not intersect each other and so $D_{\ell}\cap D\neq\emptyset$. We know that previously $D_i\cap D\neq\emptyset$ and $D_i\cap D_j\neq\emptyset$. Since $D_j\cap D_k=\emptyset$ and $D_k\cap D_{\ell}=\emptyset$, we have $D_j\cap D_{\ell}=\emptyset$. Thus, the disks $D_j, D_k, D_{\ell}$ are now pairwise independent. In this case, the disks $D_i, D_j, D_k, D_{\ell}, D$ together contradict Lemma~\ref{lem-3ind}.

\item $D'=D_{\ell}$. Proof is analogous to the case $D'=D_j$. Since $D_i\cap D\neq \emptyset$, by Lemma~\ref{lem-ind}, (i) either $D_i \cap D_j\neq\emptyset$ or $D_j\cap D\neq\emptyset$ as well as (ii) either $D_i\cap D_{k}\neq\emptyset$ or $D_{k}\cap D\neq\emptyset$. Similarly, since $D_i\cap D_{\ell}\neq\emptyset$, by Lemma~\ref{lem-ind}, (i) either $D_i \cap D_{j}\neq\emptyset$ or $D_{j} \cap D_{\ell}\neq\emptyset$. Suppose that $D_i \cap D_k\neq\emptyset$. Now, if $D_{i} \cap D_j\neq\emptyset$, then $\{D_i,D_j,D_k,D_{\ell}\}$ forms a $K_{1,3}$. A  contradiction to  Lemma~\ref{lem-three}. Therefore, $D_{i}\cap D_j=\emptyset$; consequently, $D_j\cap D_{\ell}\neq\emptyset$ and $D_{\ell}\cap D\neq\emptyset$. Then, the disks $D_j, D_{\ell}, D$ form a triangle. A contradiction. So, $D_i$ and $D_j$ should not intersect each other. Then, $D_j\cap D_{\ell}\neq\emptyset$ and $D_{\ell}\cap D\neq\emptyset$. Consequently, the disks $ D_j, D_{\ell}, D $ form a triangle. This contradicts the fact the subgraph induced by $ \Psi_{j,k,\ell} $ is triangle-free.

\end{itemize}

\noindent{\bf Case 2: $\{D, D'\}\cap\{D_j, D_k, D_{\ell}\}=\emptyset$}. Since $ D_i \cap D \neq \emptyset $, by Lemma~\ref{lem-ind}, either $D_i \cap D_j\neq\emptyset$ or $D_j\cap D\neq\emptyset$. Assume that $D_j\cap D\neq\emptyset$. If $D_j$ and $D'$ intersects, then we get a triangle formed by $D_j, D,D'$. This contradicts the definition of $B[j,k,\ell]$. But, since $D_i\cap D'\neq\emptyset$, by Lemma~\ref{lem-ind}, either $D_i\cap D_j\neq\emptyset$ or $D_j\cap D'\neq\emptyset$. We have already shown that $D_j\cap D'=\emptyset$ and so $D_i\cap D_j\neq\emptyset$. Then, the disks $\{D_i, D_j, D\}$ forms a triangle. Observe that this is Case 1 and so $D_j\cap D=\emptyset$. By a similar argument, one can show that $D_k\cap D=\emptyset$ and $D_{\ell}\cap D=\emptyset$. By Lemma~\ref{lem-ind}, since $D_i\cap D\neq\emptyset$, $D_i\cap D_j\neq\emptyset, D_i \cap D_k \neq \emptyset, D_i \cap D_{\ell} \neq \emptyset $. This implies that $ \{D_j, D_k, D_{\ell}\} $ are pairwise independent. Then $\{D_i, D_j, D_k, D _{\ell}\}$ forms a $K_{1,3}$. This contradicts to Lemma~\ref{lem-three}.
\end{proof}

\paragraph{Complexity} For each triple $(i,j,k)$ we find $ B[i,j,k] $. Then we output one that maximizes the size. To compute $B[i,j,k]$, we have to compare $O(n)$ subproblems. So the total time complexity is $O(n^4)$. Since table $B$ is of size $O(n^3) $, the total space complexity is $O(n^3)$. Therefore, we have the following theorem.
\begin{theorem}
\label{theo-2f}
Let $\mathcal{D}$ be a set of $n$ unit-disks intersecting a  straight line where all the centers lie on one side of the line. In $O(n^4) $ time and $O(n^3) $ space, we can find maximum sized $\mathcal{D'} \subseteq  \mathcal{D}$  such that  $\mathcal{D'}$ induces a bipartite subgraph. 
\end{theorem}


We now consider the $\mbs$~problem when the set of disks in $\mathcal{D}$ intersect the $x$-axis (denoted by $L$ afterwards) from both sides; that is, the centres can lie on both sides of the $x$-axis. Observe that here the corresponding graph  $G_{\mathcal{D}}$ might have induced cycles of arbitrary length (see Figure~\ref{fig-cycle13} for an example).

\begin{figure}[t]
\centering
\includegraphics[width=4in,height=1in]{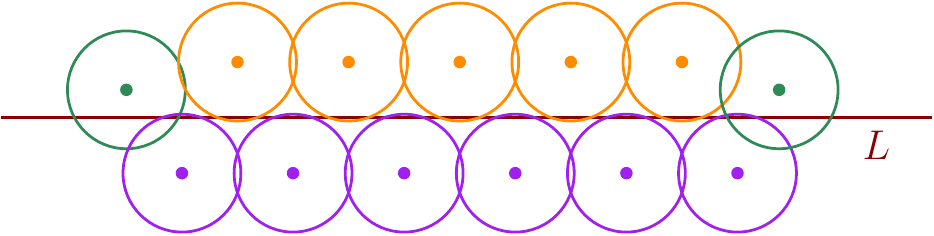}
\caption{An induced cycle of 13 disks intersecting the line $L$.} 
\label{fig-cycle13}
\end{figure}
	
We give an $O(n+|E(G_{\mathcal{D}})|)$-time 2-approximation algorithm for the $\mbs$~problem, where $ E(G_{\mathcal{D}}) $ is the set of all edges in $G_{\mathcal{D}}$. To this end, we partition the set $\mathcal{D}$ into two sets $\mathcal{D}_a$ and $\mathcal{D}_b$, where $ \mathcal{D}_a $ is the set of all disks in $\mathcal{D}$ whose centres lie above or on $L$ and $\mathcal{D}_b$ is the set of all disks in $\mathcal{D}$ whose centers are strictly below $L$. Let $G_{\mathcal{D}_a}$ (resp., $G_{\mathcal{D}_b}$) be the intersection graph of disks in $\mathcal{D}_a$ (resp., $\mathcal{D}_b$). A co-comparability graph is a graph  whose complement admits a transitive orientation. Now, one can verify that both graphs $G_{\mathcal{D}_a}$ and $G_{\mathcal{D}_b}$ are co-comparability graphs. Kohler et al.~\cite{kohler2016linear} gave an $O(|V|+|E|)$-time algorithm to compute a maximum-weighted independent set on a co-comparability graph $G=(V,E)$. Let $\mis(G_{\mathcal{D}_a})$ and $\mis(G_{\mathcal{D}_b})$ be the maximum independent sets on $G_{\mathcal{D}_a}$ and $G_{\mathcal{D}_b}$, respectively. Clearly, the vertices in $\mis(G_{\mathcal{D}_a})\cup\mis(G_{\mathcal{D}_b})$ induce a bipartite subgraph of $G_{\mathcal{D}}$. Notice that $|\mis(G_{\mathcal{D}})|\leq |\mis(G_{\mathcal{D}_a})|+ |\mis(G_{\mathcal{D}_b})|$. Let $\mbs(G_{\mathcal{D}})$ be an optimal solution for the $\mbs$ problem on $G_{\mathcal{D}}$. Since $|\mbs(G_{\mathcal{D}})|\leq 2|\mis(G_{\mathcal{D}})|$, $|\mbs(G_{\mathcal{D}})|\leq 2(|\mis(G_{\mathcal{D}_a})|+ |\mis(G_{\mathcal{D}_b})|)$. Hence, $\mis(G_{\mathcal{D}_a}) \cup \mis(G_{\mathcal{D}_b})$ gives a factor-2 for the $\mbs$ problem on $G_{\mathcal{D}}$. Therefore, we have the following theorem.
\begin{theorem}
\label{thm:2factor}
Let $\mathcal{D}$ be a set of $n$ unit disks intersecting a line in the plane and let $G_{\mathcal{D}} $ be the intersection graph of $\mathcal{D}$. An induced bipartite subgraph of size at least $|\mbs(G_{\mathcal{D}})|/2$ can be computed in $O(n+|E(G_{\mathcal{D}})|)$ time, where $\mbs(G_{\mathcal{D}})$ is an optimal solution for the $\mbs$ problem on $G_{\mathcal{D}}$.
\end{theorem}

\subsection{Arbitrary unit disks}
\label{sec3.1}
We now remove the restriction of disks intersecting a line, and consider the $\mbs$ problem on arbitrary unit disks. We first give an $O(n^4)$-time 3-approximation algorithm and will then discuss an $O(\log n)$-approximation algorithm that runs in $O(n\log n)$ time.

\paragraph{A 3-approximation algorithm} Our approach is motivated by the technique of  Agarwal et al. \cite{agarwal1998label}, where they gave a $O(n\log n)$-time 2-approximation algorithm for the maximum independent set problem on a set of $n$ unit-height rectangles in the plane. Consider a set $ \mathcal{D} $ of $n$ unit-disks in the plane. We first place a set of  horizontal lines such that each disk is intersected by some line.

Suppose there are $k \leq n$ such horizontal lines, unit distance apart, that stab all the disks. Let $\{ L_1, L_2, \dots , L_k \}$ be the set of these lines and they partition the set $ \mathcal{D} $ into subsets $\mathcal{D}_1, \mathcal{D}_2, \dots, \mathcal{D}_k,$ where $\mathcal{D}_i \subseteq \mathcal{D}$ be the set of those disks that are intersected by the line $L_i$ and centers are lying above or on $L_i$. Now we have the following observation.

\begin{observation}\label{lem-sep}
Let $ D \in \mathcal{D}_i$ and  $ D' \in \mathcal{D}_j $ be two disks in $ \mathcal{D} $, where $1 \leq i,j \leq k $. If $ |i-j| \geq 2$ then $ D$ and $ D' $ do not intersect.
\end{observation}

We now apply the algorithm to solve the $\mbs$~problem on each $\mathcal{D}_i$ ($1\leq i\leq k$). Thus, we compute a maximum bipartite subgraph $ B_i $ on each $\mathcal{D}_i$ in $O(|\mathcal{D}_i|^4)$ time. Consider the three bipartite subgraphs $\{B_1\cup B_4,\dots,B_{3\floor*{k/3}+1}\}, \{B_2\cup B_5,\dots,B_{3\floor*{k/3}-1}\}$ and $\{B_3 \cup B_6, \dots,B_{3\floor*{k/3}}\}$. Let $OPT$ denote a maximum bipartite subgraph on the graph induced by the objects in $\mathcal{D}$. It is easy to observe that $\sum_{i=1}^k |B_i|\geq |OPT|$. So, the largest set among these three must have size at least $|OPT|/3$. Thus, we obtain a 3-approximation algorithm. To find the lines $L_i$ (and form the corresponding partition), we need to do a single pass through the disks after sorting them by the $y$-coordinates of their centres. Thus, running time of the algorithm is bounded by $O(n^4)$. We hence have the following theorem.
\begin{theorem}
Let $\mathcal{D}$ be a set of $n$ unit disks in the plane and let $G_{\mathcal{D}} $ be the intersection graph of $\mathcal{D}$. An induced bipartite subgraph of size at least $|\mbs(G_{\mathcal{D}})|/3$ can be computed in $O(n^4)$ time, where $\mbs(G_{\mathcal{D}})$ is a maximum induced bipartite subgraph of $G_{\mathcal{D}}$.
\end{theorem}

\paragraph{An $O(\log n)$-approximation algorithm} Here, we describe an $O(\log n)$-approximation algorithm for $\bsu$~ problem that runs in $O(n\log n)$ time. This algorithm is also motivated by  Agarwal et al.~\cite{agarwal1998label} who gave an $O(n\log n)$-time $O(\log n)$-approximation algorithm for the maximum independent set problem for a set of $n$ axis-parallel rectangles in the plane.

Let $\mathcal{D}=\{D_1, D_2, \dots D_n \}$ and let $ c_i $ be the centers of the disks $ D_i , 1 \leq i\leq n$. We sort the centres of the disks in $\mathcal{D}$ by their $x$-coordinates. This takes $O(n \log n)$ time. If $n\leq 2$, we solve problem in constant time. Otherwise, we do the following.
\begin{enumerate}
\item Let $ c_{med} $ be the centre having the median $x$-coordinate among all the centres.
\item Draw the vertical line $x=x(c_{med})$.
\item Partition the disks $\mathcal{D}$ into three sets $\mathcal{D_{\ell}}$, $\mathcal{D}_{med}$, and $\mathcal{D}_r$ defined by
\begin{description}
	\item[] $\mathcal{D_{\ell}}=\{D_i\colon D_i\in\mathcal{D},x(c_{med})-x(c_i)>1\}$,
	\item[] $\mathcal{D}_{med}= \{D_i\colon D_i\in\mathcal{D},|x(c_{med})-x(c_i)|\leq 1\}$,
	\item[] $\mathcal{D}_r=\{D_i \colon D_i\in \mathcal{D},x(c_i)-x(c_{med})>1\}$.
\end{description}
\item Compute $ B_{med}$, a 2-approximation solution for the $\mbs$~problem on disks in $\mathcal{D}_{med}$ using Theorem~\ref{thm:2factor}.
\item Recursively compute $B_{\ell}$ and $B_{r}$, the approximate maximum bipartite subgraph induced by $\mathcal{D_{\ell}}$ and $\mathcal{D}_r$, respectively.
\item Return $B_{med}$ if $|B_{med}|\geq |B_{\ell}\cup B_{r} |$; otherwise, return $B_{\ell}\cup B_{r}$.
\end{enumerate}

Observe that for every pair of disks  $ D \in  \mathcal{D_{\ell}} $ and $D' \in  \mathcal{D}_r$, $ D \cap D' = \emptyset. $ This implies that $ B_{\ell} \cup  B_{r} $ induces a bipartite subgraph of $ G_{\mathcal{D}} $, intersection graph of $ \mathcal{D} $. By Theorem~\ref{thm:2factor}, we can obtain a 2-approximation solution for the problem on $\mathcal{D}_{med}$ in linear time. Since $|\mathcal{D_{\ell}}|\leq n/2$ and $|\mathcal{D}_{r}|\leq n/2$, the overall running time is $O(n\log n+|E|)$, where $E$ is the set of edges in $G_{\mathcal{D}}$.

Next, we show that our algorithm computes a bipartite subgraph of size at least
$|OPT|/\max(1, 2\log n)$, where $OPT$ is a maximum bipartite subgraph in $G_{\mathcal{D}}$. The proof is by induction on $n$. For the base case, the claim is true for $n\leq 2$ since we can compute a maximum bipartite subgraph in constant time. For inductive step, suppose the claim is true for all $k < n$. Let $B^*_{\ell}, B^*_{r}$ and $B^*_{med}$ denote the maximum bipartite subgraph on $G_\mathcal{D_{\ell}}, G_{\mathcal{D}_{r}}$ and $G_{\mathcal{D}_{med}}$, respectively. Since we compute a 2-approximation solution for the problem on $G_{\mathcal{D}_{med}}$, we have 
\begin{center}
$|B_{med}|\geq\dfrac{|B^*_{med}|}{2}\geq\dfrac{|OPT \cap\mathcal{D}_{med}|}{2}$.
\end{center}

By the induction hypothesis,
\begin{center}
$|B_{\ell}| \geq \dfrac{|B^*_{\ell}|}{2 \log (n/2)} \geq \dfrac{|OPT \cap \mathcal{D}_{\ell}|}{2 (\log n -1)}$ and similarly $|B_r| \geq  \dfrac{|OPT \cap \mathcal{D}_r|}{2 (\log n -1)}$.
\end{center}

Therefore,
\begin{center}
$|B_{\ell}| + |B_r| \geq \dfrac{|OPT \cap \mathcal{D}_{\ell}|+|OPT \cap \mathcal{D}_r|}{2 (\log n -1)}$
\\ $~~~~~~~~~~~= \dfrac{|OPT| - |OPT \cap \mathcal{D}_{med}| }{2 (\log n -1)}$
\end{center}
	
If $|OPT \cap \mathcal{D}_{med}| \geq |OPT|/ \log n $, then $|B_{med} | \geq \dfrac{|OPT|}{2\log n} $ and we are done.\\ Otherwise, $ |OPT \cap \mathcal{D}_{med}| < |OPT|/ \log n  $ in which case
\begin{center}
$|B_{\ell}| + |B_r| \geq  \dfrac{|OPT| - |OPT \cap \mathcal{D}_{med}| }{2 (\log n -1)}$
\\$ ~~~~~~~~~~~~~\geq  \dfrac{|OPT| - |OPT|/ \log n }{2 (\log n -1)} $
\\$ = \dfrac{|OPT|}{2\log n}$
\end{center}
	
Therefore, $\max\{|B_{med}|,|B_{\ell}|+|B_r|\} \geq |OPT|/(2\log n)$; hence, proving the induction step.
\begin{theorem}
Let $\mathcal{D}$ be a set of $n$ unit disks in the plane and let $G_{\mathcal{D}}=(V,E)$ be the intersection graph of $\mathcal{D}$. An induced bipartite subgraph of size at least $|\mbs(G_{\mathcal{D}})|/(2\log n)$ can be computed in $ O(n\log n +|E|)$ time, where $\mbs(G_{\mathcal{D}})$ denotes a maximum induced bipartite subgraph in $G_{\mathcal{D}}$.
\end{theorem}

\section{$\np$-hardness of $\tfs$}
\label{sec:tfs}
In this section, we show that $\tfs$~problem is $\np$-hard when geometric objects are axis-parallel rectangles. We give a reduction from the independent set problem on 3-regular planar graphs, which is known to be $\np$-complete~\cite{garey2002computers}.

Rim et al.~\cite{rim1995rectangle} proved that $\mis$ is $\np$-hard for planar rectangle intersection graphs with degree at most 3. They also gave a reduction from the independent set problem on 3-regular planar graphs. Given a 3-regular planar graph $G=(V,E)$, they  construct an instance $H=(V', E')$ of the $\mis$~problem on  rectangle intersection graphs.   First we outline their construction of $H$ from $G$. For any cubic planar graph $G$, it is always possible to get a rectilinear planar embedding of $G$ such that each vertex $v \in V$ is drawn as a point $p_v$, and each edge $e=(u,v)\in E$ is drawn as a rectilinear path, connecting the points $p_u$ and $p_v$, having at most four bends, and thus consisting of at most five straight line segments. They \cite{rim1995rectangle} construct a family of rectangles $B$ in the following way. For each point $p_{v_i}$ where $v_i \in V$, a rectangle $b_i$ is placed surrounding the point $p_{v_i}$. In each rectilinear path connecting $p_{v_i}$ and $p_{v_j}$, they place six rectangles $ b^1_{ij}, b^2_{ij}, \dots,b^6_{ij}  $ such that i) $b_i  $ intersects $  b^1_{ij}$, ii) $b_j  $ intersects $  b^6_{ij}$, iii) $  b^k_{ij}$ intersects $  b^{k+1}_{ij}$ for $k= 1,2, \dots ,5$ iv) $ b^1_{ij}, b^2_{ij}, \dots,b^6_{ij}  $ do not intersect any other rectangles in $ B $. For an illustration see  Figure \ref{fig-H}.
	
\begin{figure}[ht!]
\subfigure[]
{
\includegraphics[scale=0.35]{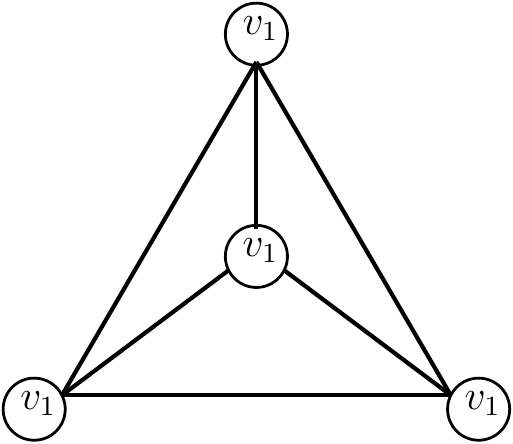}
\label{}
}	
\subfigure[]
{
\includegraphics[scale=0.40]{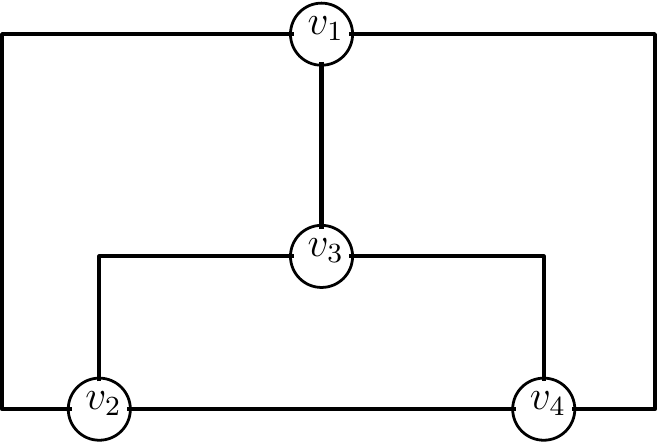}
\label{}
}
\subfigure[]
{
\includegraphics[scale=0.35]{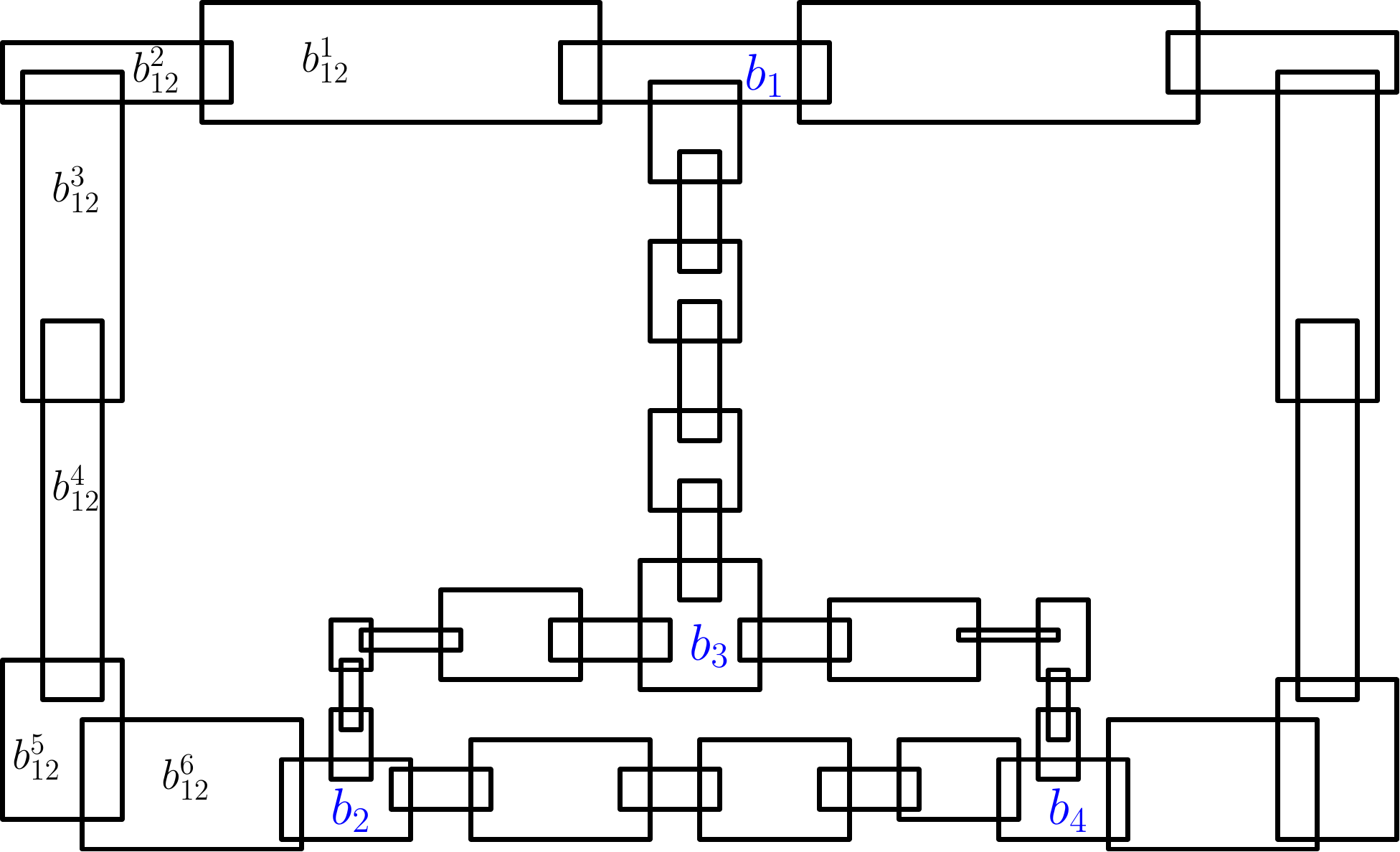}
\label{}
}
\caption{(a) A cubic planar graph $G$. (b) A rectilinear embedding of $G$. (c) Family of rectangle $B$.}
\label{fig-H}
\end{figure}

Clearly, $H(V', E')$ is an axis-parallel rectangle intersection graph with degree at most 3 where $|V'|=|V|+6|E| $ and $|E'|= 7|E|$. In their reduction, the following lemma holds.
\begin{lemma}\cite{rim1995rectangle}
\label{lem-G}
$G=(V,E)$ has an independent set of size $\geq m$ if and only if   $H$ has an independent set of size $m+3|E|$.
\end{lemma}
	
Given $H$, we construct an instance $G_H$ of $\tfs$~problem in axis-parallel rectangles intersection graphs. For the sake of understanding, let all	rectangles corresponding to vertices in $H$, i.e., all rectangles in $B$, be colored black. To get $G_H$, we insert a family $R$ of red rectangles in the following way. For each pair of adjacent rectangles $b$ and $b'$ in $B$, we place a red rectangle $R_{b,b'}$ such that i) $R_{b,b'} $ intersects both $  b$ and $b'$, ii) $R_{b,b'} $ does not intersect any other rectangles in $ B \cup R$. As per construction of $H$, it is always possible to place such a rectangle for each pair of adjacent rectangles in $B$. See Figure \ref{fig-G_H} for an illustration of this transformation. This completes the construction of our instance of the $\tfs$~problem on axis-parallel rectangles intersection graphs. Since $R$ contains $7|E|$ rectangles, the above transformation can be done in polynomial time. Now $G_H$ is an axis-parallel rectangle intersection graph with underlying geometric objects $B \cup R$. Clearly the number of vertices in $G_H$ is $(|V|+13|E|)$. We now prove the following lemma.
	
\begin{figure}[t]
\centering
\includegraphics[width=4in,height=2in]{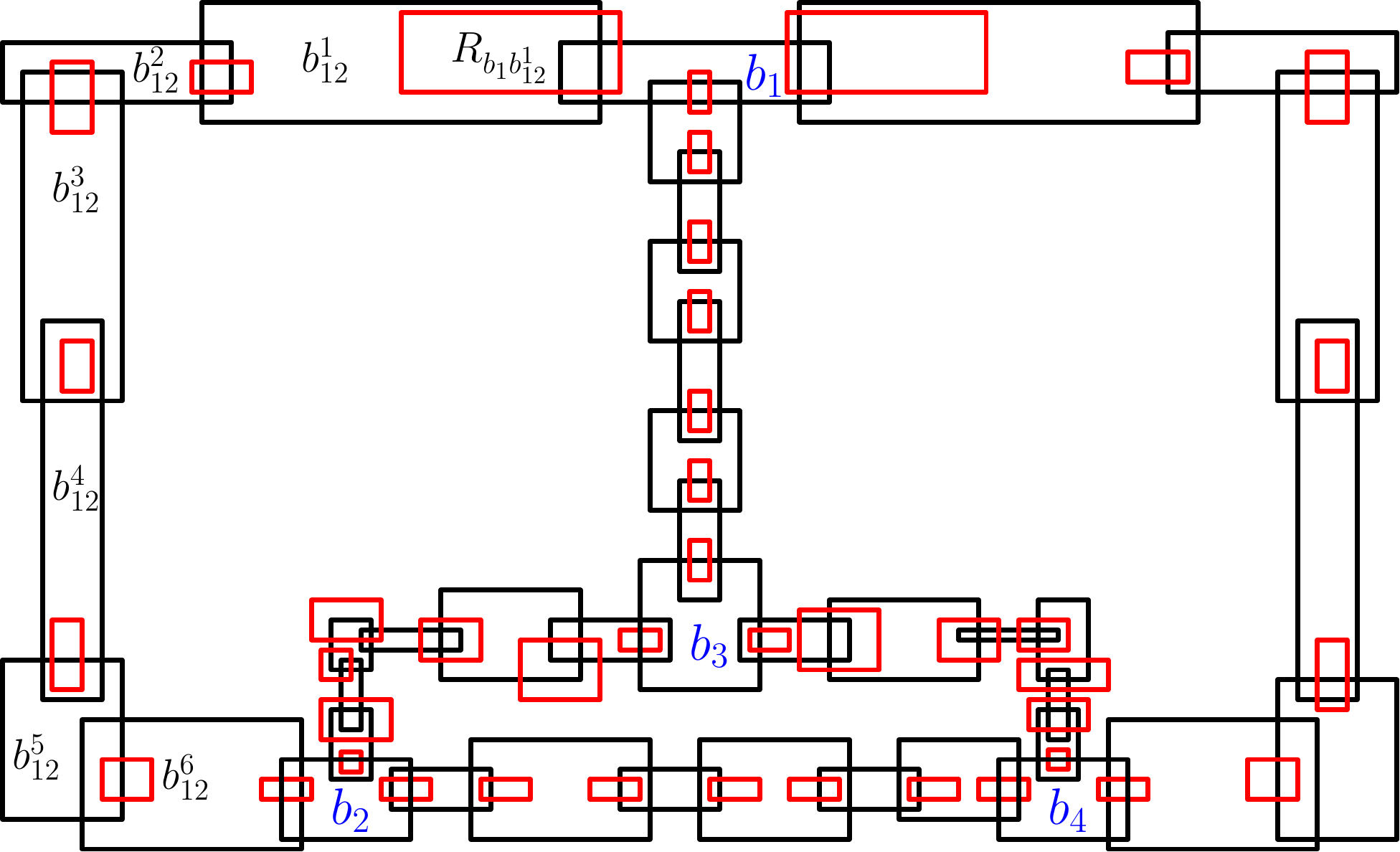}
\label{}
\caption{Construction of an instance of $\tfs$~problem from $B$.} 
\label{fig-G_H}
\end{figure}
	
\begin{lemma}
\label{lem-GH}
$H$ has an independent set of size $\geq k$ if and only if   $G_H$ has a triangle-free subgraph on $\geq k+7|E|$ vertices.
\end{lemma}
\begin{proof}
Let $H$ have an independent set of $k$ vertices. Let $B' \subseteq B$ be the set of rectangles corresponding to these $k$ vertices. Note that the rectangles in $R$ are independent. We take the subgraph $H'$ of $G_H$ induced by the vertices corresponding to the rectangles $B' \cup R$. Now $H'$ is triangle-free.  Because if $H'$ is not triangle-free then there is an intersection between two black rectangles in $B'$, that leads to a contradiction. As $|R|=7|E|$, so the claim holds. 

Now we show the other direction. Let $G_H$  have a triangle-free subgraph $H'$ on $\geq k+7|E|$ vertices. Let $X(H')= B_1 \cup R_1$ be the set of rectangles corresponding to the vertices of $H'$, where $B_1 \subseteq B$ and $R_1 \subseteq R$ . Also, let $R_2 \subseteq (R \setminus R_1)$ be the set of those red rectangles that has at most one adjacent  rectangle in $B_1$. Then clearly the graph with underlying rectangles $B_1 \cup R_1 \cup R_2 $ has no triangles. So each rectangle in $R \setminus (R_1 \cup R_2)$ has exactly two neighbours in $B_1$. Let $E(B_1)$ denote the set of edges in the subgraph induced by the vertices corresponding to the rectangles in $B_1$. Note that if there is a pair of adjacent rectangles $b_i$ and $b_j$ in $B_1$ then the rectangle $R_{i,j}$ should be part of $R_3$. It implies that $|E(B_1)|= |R_3|$ so $|E(B_1)|= |R|-|R_1 \cup R_2| $. Now $|B_1| + | R_1 \cup R_2| \geq k+ 7|E|$, so $|B_1| - (7|E|- | R_1 \cup R_2|) \geq k$. This  implies $|B_1| - (|R|- | R_1 \cup R_2|) \geq k$, hence $|B_1| - |E(B_1)| \geq k$. Now in $B_1$ we repeatedly remove the rectangles from $B_1$ to get an independent set of rectangles. Finally, within at most $|E(B_1)|$ steps, we are left with a set of independent rectangles in $B_1$ with size $\geq k$. So there exists an  independent set of size $\geq k$ in $H$.\hfill
\end{proof}

By Lemma~\ref{lem-G} and Lemma~\ref{lem-GH}, we have the following.
\begin{lemma}\label{lem-H}
$G=(V,E)$ has an independent set of size $\geq m$ if and only if $G_H$ has a triangle-free subgraph on $m+10|E|$ vertices.
\end{lemma}

We can now conclude the following theorem.
\begin{theorem}\label{thm-tfs-np-hard} 
The	$\tfs$~problem is $\np$-complete on axis-parallel rectangle intersection graphs.
\end{theorem}

\section{Conclusion}
In this paper, we studied the problem of computing a maximum-size bipartite subgraph on geometric intersection graphs. We showed that the problem is $\np$-hard on the geometric graphs for which maximum independent set is $\np$-hard. On the positive side, we gave polynomial-time algorithms for solving the problem on interval graphs and circular-arc graphs. We furthermore obtained several approximation algorithms for the problem on unit squares, unit disks, and unit-height rectangles. We also considered the $\mbs$~problem on some variants of unit-disk graphs. Finally, we showed the $\np$-hardness of a simpler problem in which the goal is to compute a maximum-size induced triangle-free subgraph. We conclude by the following open questions:
\begin{itemize}
\item[$\bullet$] Does $\mbs$ admit a $\ptas$ on unit-height rectangles, or is it $\apx$-hard?
\item[$\bullet$] Is there a polynomial-time algorithm for $\mbs$ on a set of $n$ unit disks intersecting a common horizontal line?
\end{itemize}

\paragraph{Acknowledgement} We thank Michiel Smid for useful discussions on the problem.

\bibliography{mybib}

\begin{thebibliography}{10}
\expandafter\ifx\csname url\endcsname\relax
  \def\url#1{\texttt{#1}}\fi
\expandafter\ifx\csname urlprefix\endcsname\relax\def\urlprefix{URL }\fi
\expandafter\ifx\csname href\endcsname\relax
  \def\href#1#2{#2} \def\path#1{#1}\fi

\bibitem{yannakakis1981node}
M.~Yannakakis, Node-deletion problems on bipartite graphs, {SIAM} J. Comput.
  10~(2) (1981) 310--327.

\bibitem{garey2002computers}
M.~R. Garey, D.~S. Johnson, Computers and intractability, Vol. 174, freeman San
  Francisco, 1979.

\bibitem{brandstadt1992improved}
A.~Brandst{\"{a}}dt, On improved time bounds for permutation graph problems,
  in: Graph-Theoretic Concepts in Computer Science, 18th International
  Workshop, {WG} '92, Wiesbaden-Naurod, Germany, June 19-20, 1992, Proceedings,
  1992, pp. 1--10.

\bibitem{daniel1997minimum}
Y.~D. Liang, M.~Chang, Minimum feedback vertex sets in cocomparability graphs
  and convex bipartite graphs, Acta Inf. 34~(5) (1997) 337--346.

\bibitem{kratsch2008feedback}
D.~Kratsch, H.~M{\"{u}}ller, I.~Todinca, Feedback vertex set on at-free graphs,
  Discrete Applied Mathematics 156~(10) (2008) 1936--1947.

\bibitem{honma2016algorithm}
H.~Honma, Y.~Nakajima, A.~Sasaki, An algorithm for the feedback vertex set
  problem on a normal {H}elly circular-arc graph, Journal of Computer and
  Communications 4~(08) (2016) 23.

\bibitem{DBLP:journals/siamdm/ChoiNR89}
H.~Choi, K.~Nakajima, C.~S. Rim, Graph bipartization and via minimization,
  {SIAM} J. Discrete Math. 2~(1) (1989) 38--47.

\bibitem{DBLP:conf/stoc/Yannakakis78}
M.~Yannakakis, Node- and edge-deletion np-complete problems, in: Proceedings of
  the 10th Annual {ACM} Symposium on Theory of Computing, May 1-3, 1978, San
  Diego, California, {USA}, 1978, pp. 253--264.

\bibitem{DBLP:journals/siamcomp/Hadlock75}
F.~Hadlock, Finding a maximum cut of a planar graph in polynomial time, {SIAM}
  J. Comput. 4~(3) (1975) 221--225.

\bibitem{aoshima1977comments}
K.~Aoshima, M.~Iri, Comments on f. hadlock's paper: "finding a maximum cut of a
  planar graph in polynomial time", {SIAM} J. Comput. 6~(1) (1977) 86--87.

\bibitem{DBLP:journals/siamdm/BaiouB16}
M.~Ba{\"{\i}}ou, F.~Barahona, Maximum weighted induced bipartite subgraphs and
  acyclic subgraphs of planar cubic graphs, {SIAM} J. Discrete Math. 30~(2)
  (2016) 1290--1301.

\bibitem{DBLP:journals/siamdm/CornazM07}
D.~Cornaz, A.~R. Mahjoub, The maximum induced bipartite subgraph problem with
  edge weights, {SIAM} J. Discrete Math. 21~(3) (2007) 662--675.

\bibitem{DBLP:journals/orl/ReedSV04}
B.~A. Reed, K.~Smith, A.~Vetta, Finding odd cycle transversals, Oper. Res.
  Lett. 32~(4) (2004) 299--301.

\bibitem{DBLP:conf/iwoca/LokshtanovSS09}
D.~Lokshtanov, S.~Saurabh, S.~Sikdar, Simpler parameterized algorithm for
  {OCT}, in: Combinatorial Algorithms, 20th International Workshop, {IWOCA}
  2009, Hradec nad Moravic{\'{\i}}, Czech Republic, June 28-July 2, 2009,
  Revised Selected Papers, 2009, pp. 380--384.

\bibitem{DBLP:conf/fsttcs/LokshtanovSW12}
D.~Lokshtanov, S.~Saurabh, M.~Wahlstr{\"{o}}m, Subexponential parameterized odd
  cycle transversal on planar graphs, in: {IARCS} Annual Conference on
  Foundations of Software Technology and Theoretical Computer Science, {FSTTCS}
  2012, December 15-17, 2012, Hyderabad, India, 2012, pp. 424--434.

\bibitem{DBLP:journals/jacm/HochbaumM85}
D.~S. Hochbaum, W.~Maass, Approximation schemes for covering and packing
  problems in image processing and {VLSI}, J. {ACM} 32~(1) (1985) 130--136.

\bibitem{DBLP:conf/soda/ChalermsookC09}
P.~Chalermsook, J.~Chuzhoy, Maximum independent set of rectangles, in:
  Proceedings of the Twentieth Annual {ACM-SIAM} Symposium on Discrete
  Algorithms, {SODA} 2009, New York, NY, USA, January 4-6, 2009, 2009, pp.
  892--901.

\bibitem{DBLP:conf/wads/BoseCKMMMS19}
P.~Bose, P.~Carmi, J.~M. Keil, A.~Maheshwari, S.~Mehrabi, D.~Mondal, M.~H.~M.
  Smid, Computing maximum independent set on outerstring graphs and their
  relatives, in: Algorithms and Data Structures - 16th International Symposium,
  {WADS} 2019, Edmonton, AB, Canada, August 5-7, 2019, Proceedings, 2019, pp.
  211--224.

\bibitem{lu1997linear}
C.~L. Lu, C.~Y. Tang, A linear-time algorithm for the weighted feedback vertex
  problem on interval graphs, Inf. Process. Lett. 61~(2) (1997) 107--111.

\bibitem{clark1990unit}
B.~N. Clark, C.~J. Colbourn, D.~S. Johnson, Unit disk graphs, Discrete
  Mathematics 86~(1-3) (1990) 165--177.

\bibitem{kratochvil1990independent}
J.~Kratochv{\'\i}l, J.~Ne{\v{s}}et{\v{r}}il, Independent set and clique
  problems in intersection-defined classes of graphs, Commentationes
  Mathematicae Universitatis Carolinae 31~(1) (1990) 85--93.

\bibitem{lahiri2015maximum}
A.~Lahiri, J.~Mukherjee, C.~R. Subramanian, Maximum independent set on
  {B}{$_1$}-{VPG} graphs, in: Combinatorial Optimization and Applications - 9th
  International Conference, {COCOA} 2015, Houston, TX, USA, December 18-20,
  2015, Proceedings, 2015, pp. 633--646.

\bibitem{DBLP:series/mcs/DowneyF99}
R.~G. Downey, M.~R. Fellows, Parameterized Complexity, Monographs in Computer
  Science, Springer, 1999.

\bibitem{DBLP:conf/esa/Marx05}
D.~Marx, Efficient approximation schemes for geometric problems?, in:
  Algorithms - {ESA} 2005, 13th Annual European Symposium, Palma de Mallorca,
  Spain, October 3-6, 2005, Proceedings, 2005, pp. 448--459.

\bibitem{marx2006parameterized}
D.~Marx, Parameterized complexity of independence and domination on geometric
  graphs, in: Parameterized and Exact Computation, Second International
  Workshop, {IWPEC} 2006, Z{\"{u}}rich, Switzerland, September 13-15, 2006,
  Proceedings, 2006, pp. 154--165.

\bibitem{DBLP:conf/jcdcg/Matsui98}
T.~Matsui, Approximation algorithms for maximum independent set problems and
  fractional coloring problems on unit disk graphs, in: Discrete and
  Computational Geometry, Japanese Conference, JCDCG'98, Tokyo, Japan, December
  9-12, 1998, Revised Papers, 1998, pp. 194--200.

\bibitem{DBLP:journals/ipl/NandyPR17}
S.~C. Nandy, S.~Pandit, S.~Roy, Faster approximation for maximum independent
  set on unit disk graph, Inf. Process. Lett. 127 (2017) 58--61.

\bibitem{kohler2016linear}
E.~K{\"{o}}hler, L.~Mouatadid, A linear time algorithm to compute a maximum
  weighted independent set on cocomparability graphs, Inf. Process. Lett.
  116~(6) (2016) 391--395.

\bibitem{agarwal1998label}
P.~K. Agarwal, M.~J. van Kreveld, S.~Suri, Label placement by maximum
  independent set in rectangles, Comput. Geom. 11~(3-4) (1998) 209--218.

\bibitem{rim1995rectangle}
C.~S. Rim, K.~Nakajima, On rectangle intersection and overlap graphs, IEEE
  Transactions on Circuits and Systems I: Fundamental Theory and Applications
  42~(9) (1995) 549--553.

\end{thebibliography}

\end{document}